\documentclass[aps,prl,twocolumn,superscriptaddress]{revtex4}
\usepackage{amsfonts,epsfig}
\usepackage{hyperref}
\usepackage{latexsym}
\usepackage{amsmath}
\usepackage{amssymb}
\usepackage{amsfonts}
\usepackage{amsthm}
\usepackage{mathrsfs}
\usepackage{color,verbatim,graphics}
\usepackage{psfrag}
\usepackage{verbatim,cleveref}

\input{Qcircuit}

\newcommand{\braket}[2]{\left \langle #1 | #2 \right \rangle}

\newcommand{\norm}[1]{\left|\left|#1\right|\right|}

\newcommand{\be}{\begin{equation}}
\newcommand{\ee}{\end{equation}}
\newcommand{\ba}{\begin{eqnarray}}
\newcommand{\ea}{\end{eqnarray}}
\newcommand{\ban}{\begin{eqnarray*}}
\newcommand{\ean}{\end{eqnarray*}}

\newtheorem{theorem}{Theorem}

\newcommand{\bracket}[3]{\langle#1|#2|#3\rangle}

\begin{document}

\title{Robust self-testing of the singlet}

\author{Matthew McKague}
\affiliation{Centre for Quantum Technologies, National University of Singapore, 3 Science drive 2, Singapore 117543}
\author{Tzyh Haur Yang}
\affiliation{Centre for Quantum Technologies, National University of Singapore, 3 Science drive 2, Singapore 117543}
\author{Valerio Scarani}
\affiliation{Centre for Quantum Technologies, National University of Singapore, 3 Science drive 2, Singapore 117543}
\affiliation{Department of Physics, National University of Singapore, 2 Science drive 3, Singapore 117542}

\begin{abstract}
In this paper, we introduce a general framework to study the concept of robust self testing which can be used to self test EPR pairs and local measurement operators. The result is based only on probabilities obtained from experiment, with tolerance to experimental errors. In particular, we show that if results of experiment come approach the Cirel'son bound, or approximates the Mayers-Yao type correlation, then the experiment must contain an approximate EPR pair. More specifically, there exist local bases in which the physical state is close to an EPR pair, possibly all encoded in a larger environment or ancilla.  Moreover, in theses bases the measurements are close to the qubit operators used to achieve the Cirel'son bound or the Mayers-Yao results.
\end{abstract}

\maketitle

\textit{Introduction. ---} It is well known by now that the correlations obtained by measuring entangled quantum systems cannot be reproduced with classical resources. In fact, for some of these correlations, a much stronger statement holds: they can be reproduced only by measuring a specific quantum state in a specific way. To date, two such examples are known for the bipartite case. One uses the Clauser-Horne-Shimony-Holt (CHSH) criterion \cite{Clauser:1969:Proposed-Experi} to state the following: if the maximal quantum value $\textrm{CHSH}=2\sqrt{2}$ \cite{Cirelson:1980:Quantum-general} is observed, then the state being measured is necessarily equivalent (in a sense to be made rigorous below) to a maximally entangled state of two qubits, which will be referred to as ``singlet" from now on. Moreover, both the measurements on Alice and the measurements on Bob must anti-commute \cite{Popescu:1992:Which-states-vi} \cite{Braunstein:1992:Maximal-violati}. The other criterion is due to Mayers and Yao: it uses a different observations to reach the same conclusion \cite{Mayers:2004:Self-testing-qu}. Since the Mayers-Yao correlations cannot reach $\textrm{CHSH}=2\sqrt{2}$, the two criteria are inequivalent. Compactly, we shall say that these two criteria realize the \textit{self-testing} of the singlet and of some measurements.

The possibility of self-testing is all the more remarkable because nothing is assumed \textit{a priori} on the physical system or on the measurements, not even the dimension of the relevant Hilbert space: in principle, these are \textit{device-independent} assessments, based only on the observed statistical data. Device-independent assessment has been discussed in various scenarios, including adversarial ones, which may provide the ultimate test of trustfulness. More realistic, and probably more relevant for today's physics, is a scenario in which neither the experimentalists nor nature are assumed to cheat, but where one wants a simple and direct check that nothing serious is going wrong, that there are no undesired side channels etc.

In order to be practical, a self-testing procedure must be \textit{robust}, i.e.\ tolerate deviations from the theoretically ideal case. A mathematical \textit{tour de force} has recently provided a robustness bound for the Mayers-Yao test \cite{Magniez:2006:Self-testing-of}. To our knowledge, no robust bound is available for the CHSH test, a situation that plagues the applicability of the corresponding device-independent assessment of entanglement of a source \cite{Bardyn:2009:Device-independ} and a measurement \cite{Rabelo:2011:Device-Independ}.

In this paper, we prove a general sufficient criterion for a set of correlations to provide robust self-testing of the singlet. Then we prove that both the CHSH and the Mayers-Yao tests satisfy this criterion and give the explicit bounds. The proofs use rather elementary quantum mathematics, following the simplification of the Mayers-Yao proof by one of us \cite{McKague:2010:Quantum-Informa,McKague:2010:Self-testing-gr,McKague:2011:Generalized-Sel}.

\textit{Definitions and notation. ---} We are aiming at self-testing the presence of a maximally entangled state of two qubits in unknown devices. This goal calls for a suitable definition. Indeed, there is nothing like an isolated qubit in nature: if one wants to measure the spin of an electron, the whole electron with its wavefunction is present; and if the qubit is the polarization of an optical mode, we are allowing the whole electromagnetic field to be present. So there will surely be degrees of freedom which do not encode the state of interest, but are nevertheless present. Also, there must be a local frame of reference for each device in order to define the measurements. Because of these two facts, our definition must allow for additional ancillas and local changes of basis. We do so by using an \emph{isometry}, that is a linear map $\Phi: \mathcal{H}_{A} \rightarrow \mathcal{H}_{B}$ that preserves inner products. As a concrete example, adding an ancilla and applying a unitary to the total system is an isometry.

Now we are ready to formalize our definition.  We say that a pair of devices $A$ and $B$ hold a pair of maximally entangled qubits if there exists a local isometry $\Phi = \Phi_{A} \otimes \Phi_{B}$ that takes the state $\ket{\psi^{\prime}}_{AB}$ to
\be
\Phi\left(\ket{\psi^{\prime}}_{AB}\right) = \ket{junk}_{AB}\ket{\phi_{+}}_{AB}
\ee
and physical observables $M^{\prime}_{A}$ and $N^{\prime}_{B}$ operate as
\begin{equation}
\Phi\left(M^{\prime}_{A}N^{\prime}_{B}\ket{\psi^{\prime}}_{AB}\right) = \ket{junk}_{AB}M_{A}N_{B} \ket{\phi_{+}}_{AB}
\end{equation}
for some $M_{A}$ and $N_{B}$ to be specified later.  In other words, we aim at checking that there exist a choice of local bases such that (i) the state looks like an ancilla tensored with a maximally entangled pair of qubits; and (ii) the measurements act non-trivially only on the pair of qubits.

A word about notation: we use primed notation ($X^{\prime}$, $\ket{\psi^{\prime}}$ etc.) to represent the observables and states in the actual quantum devices. These will be unknown (even their dimensions) except for a few properties that we shall specify. Non-primed operators $X$ and $Z$ refer to Pauli operators while the singlet is given by $\ket{\phi_{+}} = \frac{1}{\sqrt{2}}(\ket{00} + \ket{11})$.


\textit{Circuit for self-testing. ---} We start by presenting a set of sufficient conditions to self-test the singlet along with the associated measurement operators. The \textit{state} $\ket{\psi^{\prime}}$ can be taken as pure without loss of generality, since the dimension is not fixed and one can always add the ancillas for purification. We assume further that the state is always the same in each run of the experiment, which is reasonable in the non-adversarial scenario (this assumption could be removed, for instance, using Azuma's inequality as in \cite{Pironio:2009:Device-independ}; we have decided to deal with this technical complication in a full-length publication). The \textit{measurement settings} are denoted by $\{A'_0,A'_1,\ldots\}$ on Alice's side and $\{B'_0,B'_1,\ldots\}$ on Bob's side. For all that follows, it is a crucial assumption that $[A'_j,B'_k]=0$: this can ultimately be enforced by space-like separation of the measurement; but one may be less demanding and take simple spatial separation as a sufficient guarantee of commutation.

With these notations and assumptions, the following theorem holds:

\begin{theorem}\label{theorem:localisometry}
Suppose that from the observed correlations, one can deduce the existence of local observables $\{X^{\prime}_{A}$, $Z^{\prime}_{A}\}$ (functions of $A'_i$), and $\{X^{\prime}_{B}$, $Z^{\prime}_{B}\}$ (functions of $B'_i$)  with eigenvalues $\pm 1$, which act on the bipartite state $\ket{\psi^{\prime}}$ such that
\begin{eqnarray}
\label{eq:th1cond1}
||(X^{\prime}_{A} Z^{\prime}_{A} +  Z^{\prime}_{A} X^{\prime}_{A} ) \ket{\psi^{\prime}}|| & \leq & 2\epsilon_1, \\
\label{eq:th1cond2}
||(X^{\prime}_{B} Z^{\prime}_{B} +  Z^{\prime}_{B} X^{\prime}_{B} ) \ket{\psi^{\prime}}|| & \leq & 2\epsilon_1, \\
\label{eq:th1cond3}
||(X^{\prime}_{A} -  X^{\prime}_{B}) \ket{\psi^{\prime}}|| &\leq &  \epsilon_2, \\
\label{eq:th1cond4}
||(Z^{\prime}_{A} -  Z^{\prime}_{B}) \ket{\psi^{\prime}}|| &\leq &  \epsilon_2\,. 
\end{eqnarray}
Then there exists a local isometry $\Phi = \Phi_{A}  \otimes \Phi_{B}$ and a state $\ket{junk}_{AB}$ such that
\begin{equation}
\norm{ \Phi(M^{\prime}_{A} N^{\prime}_{B}\ket{\psi^{\prime}}) - \ket{junk}_{AB} M_{A}N_{B} \ket{\phi_{+}}_{AB}} \leq \varepsilon \label{eq:selftest}
\end{equation}
for $M,N \in \{I,X,Z\}$ and $\varepsilon=(11\epsilon_1 + 5\epsilon_2)/2$.

\end{theorem}

\begin{proof}
The isometry is constructed as in figure~\ref{fig:epr_local_unitary_circuit}.
\begin{figure}[htbp!]
\[
\Qcircuit @C=0.5cm @R=.3cm {
\lstick{\ket{0}}  & \gate{H} & \ctrl{1}  & \gate{H} & \ctrl{1} &\qw \\
 &  & \gate{Z^{\prime}_{A}} & \qw & \gate{X^{\prime}_{A}} & \qw\\
 \lstick{M^{\prime}_{A} N^{\prime}_{B}\ket{\psi^{\prime}}} \ar@{-}[ur] \ar@{-}[dr] & & & & & \\
              &    & \gate{Z^{\prime}_{B}} & \qw & \gate{X^{\prime}_{B}} & \qw\\
\lstick{\ket{0}}     & \gate{H} & \ctrl{-1} & \gate{H} & \ctrl{-1} & \qw
}
\]
\caption{Local isometry $\Phi$, where $M,N\in \{I,X,Z\}$}
\label{fig:epr_local_unitary_circuit}
\end{figure}
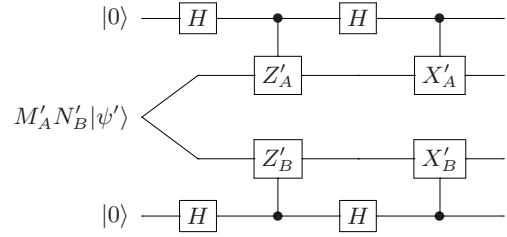
For the case where $M=N=I$,the isometry gives 
\begin{eqnarray}
	 \Phi(\ket{\psi^{\prime}}) & = & 
	 \frac{1}{4} (I + Z^{\prime}_{A})  (I + Z^{\prime}_{B}) \ket{\psi^{\prime}}\ket{00} \nonumber\\
	 & + & \frac{1}{4} X^{\prime}_{B}(I + Z^{\prime}_{A})   (I - Z^{\prime}_{B}) \ket{\psi^{\prime}}\ket{01} \nonumber\\
	 & + & \frac{1}{4}X^{\prime}_{A}(I- Z^{\prime}_{A})  (I + Z^{\prime}_{B}) \ket{\psi^{\prime}}\ket{10}\nonumber\\
	 & + & \frac{1}{4} X^{\prime}_{A}   X^{\prime}_{B} (I - Z^{\prime}_{A})(I-Z^{\prime}_{B}) \ket{\psi^{\prime}}\ket{11}.\label{eq:phiphi}
\end{eqnarray}
The rest of the proof is a series of estimates based on (\ref{eq:th1cond1}-\ref{eq:th1cond4}) and on the fact that the operators are unitary and Hermitian. The bounds are rather simple to derive (and they are obvious in the limiting case $\epsilon_1=\epsilon_2=0$); for clarity, we give only the results in the main text and provide the detailed derivation in the Supplementary Information.

In the expression for $\Phi(\ket{\psi^{\prime}})$ above, the second and third line are each bounded by $\epsilon_2/2$, while the last line differs from the first by $\epsilon_1+\epsilon_2$. From these, we have 
\begin{equation}
	\norm{
		\Phi\left(
			\ket{\psi^{\prime}}
		\right) 
	- 
		\frac{	(I + Z^{\prime}_{A})(I + Z^{\prime}_{B})}{2 \sqrt{2}} 
		\ket{\psi^{\prime}} \ket{\phi_{+}}
	} \leq 
	\epsilon_{1} + 2 \epsilon_{2}. \label{robustsoon}
\end{equation}
This is already the desired form and we would like to identify $\frac{	(I + Z^{\prime}_{A})(I + Z^{\prime}_{B})}{2 \sqrt{2}} \ket{\psi^{\prime}}$ with $\ket{junk}$; but the latter is supposed to be normalized, while the former may not be (unless $\epsilon_1=\epsilon_2=0$); so we have to estimate the error that is introduced by normalizing the state. This is found to be $(\epsilon_1 + \epsilon_2)/2$, the most tedious estimate being the one that bounds from above both $|\bracket{\psi'}{Z_A'}{\psi'}|$ and $|\bracket{\psi'}{Z_B'}{\psi'}|$ with $\epsilon_1 + \epsilon_2$. All in all therefore \begin{equation}
	\norm{
		\Phi(\ket{\psi^{\prime}})
		 - 
		 \ket{junk}\ket{\phi_{+}}
	} \leq
	\frac{3}{2} \epsilon_{1} + \frac{5}{2} \epsilon_{2}\,.
\end{equation}
This is the self-testing bound for the state. In order to derive the bound for the action of the operators, we notice that $\Phi\left(M^{\prime}_{A}  N^{\prime}_{B}  \ket{\psi^{\prime}}\right)=\frac{1}{4} (I + Z^{\prime}_{A})  (I + Z^{\prime}_{B}) M^{\prime}_{A}N^{\prime}_{B}\ket{\psi^{\prime}}\ket{00} + \textrm{(similar terms)}$. One starts by propagating $M^{\prime}_{A}$ and $N^{\prime}_{B}$ to the left using \eqref{eq:th1cond1} and \eqref{eq:th1cond2}. In the worst case, i.e.\ when both $M^{\prime}_{A}$ and $N^{\prime}_{B}$ are not the identity, this preliminary step adds $4 \epsilon_{1}$ to the bound. The resulting expression is analogous to (\ref{eq:phiphi}): then, one follows the same steps as above. 
\end{proof}

This theorem implies that there exist some local bases in which, up to some small error, the state shared by the quantum devices is a singlet together with some ancillas in an unknown state, and the derived operators $X^{\prime}$ and $Z^{\prime}$ operate only on the singlet. Additionally, in these local bases, $X^{\prime}$ and $Z^{\prime}$ are close to the Pauli $X$ and $Z$. In the remainder of the paper, we are going to show that (\ref{eq:th1cond1}-\ref{eq:th1cond4}) follow from both the CHSH and Mayers-Yao correlation experiments: therefore, both experiments can be used for robust self-testing.


\textit{Robsut self-testing using CHSH. ---} As mentioned above, a robustness bound for the CHSH-based self-testing was missing. We provide it here:

\begin{theorem}\label{theorem:chshrobust}
Suppose that the observables $A'_0$, $A'_1$, $B'_0$ and $B'_1$ with eigenvalues $\pm 1$, acting on a state $\ket{\psi^{\prime}}$, are such that
\begin{equation}
	\bra{\psi^{\prime}}\left(   A'_0B'_0 + A'_0B'_1+A'_1B'_0-A'_1B'_1     \right)  \ket{\psi^{\prime}} \geq 2 \sqrt{2} -\epsilon, \label{chsherror}
\end{equation}
where $0<\epsilon<1$. Then the conditions of theorem~\ref{theorem:localisometry} are satisfied with $\epsilon_1=2(\epsilon\sqrt{2})^{1/2}$ and $\epsilon_2=4(\epsilon\sqrt{2})^{1/4}$.
\end{theorem}

\begin{proof}
To establish the theorem, we need to show the existence of four local, Hermitian and unitary operators $X_A'$, $Z_A'$, $X_B'$ and $Z_B'$ that satisfy (\ref{eq:th1cond1})-(\ref{eq:th1cond4}). We are going to show this for
\begin{align}
	\begin{array}{ll}
	X_A' =A_0', & Z_A' = A_1', \\
	X_B' = \dfrac{B_0'+B_1'}{|B_0'+B_1'|}, & 	Z_B' = \dfrac{B_0'-B_1'}{|B_0'-B_1'|},
\end{array}
	\label{defnition}
\end{align}
where $|M| = \sqrt{M^{2}}$. Clearly they are all unitary and Hermitian \footnote{If $M$ has a subspace with eigenvalue 0, the eigenvalue of $M/|M|$ in that subspace is taken to be 1.}. Moreover, $\{X_B',Z_B'\}=0$ by construction, thus establishing a tighter version of \eqref{eq:th1cond2}. All the subsequent steps are again somehow pedestrian, so we sketch them here and leave the full details for the Supplementary Information. 

From \eqref{chsherror}, a suitable use of the Cauchy-Schwartz and the triangle inequalities leads to
\begin{eqnarray}
	\norm{\{A_0',A_1'\}\ket{\psi'}}& \leq & 2\epsilon_{1}, 	\label{eq:a0a1}\\
	\norm{\{B_0',B_1'\}\ket{\psi'}}& \leq & 2 \epsilon_{1} \label{eq:b0b1}
\end{eqnarray}
with $\epsilon_1=2\sqrt{\epsilon\sqrt{2}}$. Then \eqref{eq:th1cond1} is established in \eqref{eq:a0a1}.

The third condition \eqref{eq:th1cond3} is proved by obtaining first the bound 
$\norm{\left(X^{\prime}_{A} - (B_0' + B_1')/\sqrt{2}\right) \ket{\psi^{\prime}}}\leq 2(\epsilon\sqrt{2})^{1/4}$, then the same bound for 
$\norm{\left(X^{\prime}_{B} - (B_0' + B_1')/\sqrt{2}\right) \ket{\psi^{\prime}}}$; both derivations using \eqref{chsherror} at one point. The triangle inequality completes the estimate. The proof of \eqref{eq:th1cond4} follows the same steps.
\end{proof}

Notice that inequality~\eqref{eq:selftest} applies for $M^{\prime}_{A} = A_{0} = X^{\prime}_{A}$ and $M^{\prime}_{A} = A_{1} = Z^{\prime}_{A}$, $N^{\prime}_{B}= X^{\prime}_{B}$ and $N^{\prime}_{B} = Z^{\prime}_{B}$. One may want to have a self-testing bound for the operators that are really measured, $B'_{0}$ and $B'_{1}$, which are not linear functions of the previous ones. The inequality for $B^{\prime}_{0}$ is found by using linearity in the estimations $\norm{\left(X^{\prime}_{A} - (B_0' + B_1')/\sqrt{2}\right) \ket{\psi^{\prime}}}\leq 2(\epsilon\sqrt{2})^{1/4}$ and $\norm{\left(Z^{\prime}_{A} - (B_0' - B_1')/\sqrt{2}\right) \ket{\psi^{\prime}}}\leq 2(\epsilon\sqrt{2})^{1/4}$ from the proof, then using the fact that isometries preserve the 2-norm. We obtain
\begin{widetext}
\begin{equation}
\norm{ \Phi\left(M^{\prime}_{A} B'_{0}\ket{\psi^{\prime}}\right) - \ket{junk}_{AB} M_{A} \frac{X_{B} + Z_{B}}{\sqrt{2}}\ket{\phi_{+}}_{AB}} \leq \sqrt{2}\epsilon + 2 \sqrt{2}\left(\epsilon \sqrt{2}\right)^{\frac{1}{4}}.
\end{equation}
\end{widetext}
The analogous result holds for $B'_{1}$.


\textit{Robust self-testing using the Mayers-Yao criterion. ---} We turn now to the robustness bound for the Mayers-Yao correlations. The original scenario uses three measurements on Alice's side and three on Bob's side; as a matter of fact though, only two measurements are needed by (say) Alice, so we work in this more economic case.

We have the following theorem:

\begin{theorem}\label{thm:myrobust}
Let $0 < \epsilon <1 $ be given and let a bipartite state $\ket{\psi^{\prime}}$ and observables $X^{\prime}_{A}$, $Z^{\prime}_{A}$, $X^{\prime}_{B}$, $Z^{\prime}_{B}$, and $D^{\prime}_{B}$ with eigenvalues $\pm 1$, be given such that 
\begin{equation}
	\left|\bra{\psi^{\prime}} M^{\prime}_{A}N^{\prime}_{B} \ket{\psi^{\prime}} - \bra{\phi_{+}} M_{A}  N_{B} \ket{\phi_{+}} \right| \leq \epsilon \label{mycond}
\end{equation}
holds for all $M \in \{X,Z\}$ and $N \in \{X,Z,D\}$ where $D=(X+Z)/\sqrt{2}$. Then the conditions of theorem~\ref{theorem:localisometry} are satisfied with $\epsilon_1=2(1+\sqrt{2})(2\epsilon)^{1/4}+4\sqrt{2\epsilon}+\frac{5+3\sqrt{2}}{2}(2\epsilon)^{3/4}$ and $\epsilon_2=\sqrt{2\epsilon}$.
\end{theorem}

The proof is included in the Supplementary Information.  The critical step is to use the strong non-classical correlations between $D^{\prime}_{B}$ and both $X^{\prime}_{A}$ and $Z^{\prime}_{A}$ to establish that the latter approximately anti-commute.  Interestingly, $D^{\prime}_{B}$ is used only in this task, and does not play any role in the other estimations or the construction of the isometry.


\textit{Discussion. ---} We have presented robust self-testing bounds for the singlet from the two criteria of CHSH and Mayers-Yao. These tools can be used for the device-independent assessment of state preparation and measurement devices, in protocols that are based on these criteria. For instance, in \cite{Bardyn:2009:Device-independ} the authors defined the ``Mayers-Yao'' fidelity as a device independent state estimation parameter based on its CHSH violation. They conjectured a lower bound on this fidelity in terms of the CHSH value, giving a construction which saturates the bound.  However, no actual lower bound was given. Our robustness bound for self-testing using CHSH can be straightforwardly converted into such a lower bound, the dominant contribution to which is
\begin{equation}
	F_{MY}(\ket{\psi^{\prime}}) \gtrsim 1-\dfrac{1}{4}(9\sqrt{2}\epsilon+2^{\frac{1}{4}}100\epsilon^{\frac{1}{2}}+2^{\frac{3}{8}}60\epsilon^{\frac{3}{4}})
\end{equation}
where $\epsilon$ is defined as in (\ref{chsherror}). This bound is rather loose: the fidelity drops to $F_{MY}=20\%$ already for $\epsilon\approx 10^{-4}$.

More importantly, the framework introduced here allows one to generalize the concept of self testing in different ways. For instance, our framework is an interesting contrast to the previous device-independent work with the CHSH inequality \cite{Pironio:2009:Device-independ} \cite{Bardyn:2009:Device-independ}as we do not rely on Jordan's lemma \cite{Jordan:1875:Essai-sur-la-ge} to reduce the high dimension case to the qubit case.  Jordan's lemma only applies in the context of two measurement settings and two outcomes, which limits the applicability of the proof techniques used in these previous papers.  The current proof technique thus provides an opening for testing in scenarios with more settings or outcomes, with the Mayers-Yao scenario a concrete example.

Considering the similarities between the CHSH and Mayers-Yao experiments we use here, it is natural to ask whether they can be generalized to a larger class of experiments which can be used to self test singlets.  The framework we use here is a natural starting point for such an enquiry since it is agnostic as to the number of settings or outcomes, requiring only that a pair of anti-commuting operators can be found, or constructed as in the CHSH case.  It also naturally extends to multi-partite scenarios as in \cite{McKague:2010:Self-testing-gr}. 

{\bf Acknoledgements} We acknowledge discussions with Lana Sheridan, Rafael Rabelo, Melvyn Ho and Cai Yu. This work is funded by the Ministry of Education and the National Research Foundation, Singapore.


%

\newpage
\appendix
\section{Detailed derivation of the bounds used in Theorem \ref{theorem:localisometry}}

\begin{itemize}
\item Bound for the second line of (\ref{eq:phiphi}), the one for the third line being identical:
\ban ||(I+ Z_A')(I- Z_B' ) \ket{\psi'} ||&\leq& \\||(I- Z_A'Z_B' ) \ket{\psi'} ||+||(Z_A'-Z_B')\ket{\psi'} ||\stackrel{(\ref{eq:th1cond4})}{=}2\epsilon_2\,. \ean
\item Comparison between the first and the fourth line of (\ref{eq:phiphi}): we want to bound \ban
||X_A'X_B(I+Z_A')(I+Z_B')\ket{\psi'}-(I+Z_A')(I+Z_B')\ket{\psi'}||\,.
\ean
The trick consists in propagating $X_A'X_B'$ to the right using (\ref{eq:th1cond1}) and (\ref{eq:th1cond2}). This costs $4\epsilon_1$ and leads to 
\ban
||(I+Z_A')(I+Z_B')(X_A'X_B'-I)\ket{\psi'}||\,.
\ean Using (\ref{eq:th1cond3}), this can be replaced by zero at the cost of $4\epsilon_2$. 
\item Bound for $|\bracket{\psi'}{Z_A'}{\psi'}|$, the same holding for $|\bracket{\psi'}{Z_B'}{\psi'}|$: this proof uses routinely two arguments: (i) the fact that the operators are unitary, and (ii) the fact that if $||\ket{\varphi}||\leq \epsilon$, then $|\braket{\chi}{\varphi}|\leq \epsilon$ for all normalized $\ket{\chi}$. We need to establish two relations. From (i) and (\ref{eq:th1cond2}),
\ban ||Z_A'X_B'\ket{\psi'}-Z_A'X_A'\ket{\psi'}||\leq \epsilon_2 \ean. By inserting $0=X_A'Z_A'-X_A'Z_A'$, the triangle inequality and (\ref{eq:th1cond1}) lead to
\ban
||Z_A'X_B'\ket{\psi'}+X_A'Z_A'\ket{\psi'}||\leq 2\epsilon_1+\epsilon_2\,.
\ean Using (ii) with $\ket{\chi}=X_B'\ket{\psi'}$ and the unitarity of $X_B'$,
\ban
|\bracket{\psi'}{Z_A'}{\psi'}+\bracket{\psi'}{X_B'X_A'Z_A'}{\psi'}|\leq 2\epsilon_1+\epsilon_2\,.
\ean Finally, since the left hand side is an absolute value, the same holds for the conjugate; whence we find the first relation
\ban
|\bracket{\psi'}{Z_A'}{\psi'}+\bracket{\psi'}{Z_A'X_A'X_B'}{\psi'}|\leq 2\epsilon_1+\epsilon_2\,.
\ean
The second relation is \ban
|\bracket{\psi'}{Z_A'}{\psi'}-\bracket{\psi'}{Z_A'X_A'X_B'}{\psi'}|\leq \epsilon_2\,,
\ean obtained simply by combining (i) and (\ref{eq:th1cond3}) in the form $||Z_A'\ket{\psi'}-Z_A'X_A'X_B'\ket{\psi'}||\leq \epsilon_2$, then using (ii) with $\ket{\chi}=\ket{\psi'}$. The two relations together, by triangle inequality, imply $|\bracket{\psi'}{Z_A'}{\psi'}|\leq \epsilon_1+\epsilon_2$.

\item Bound for the norm of the state: notice first that $(1+Z_A')^2=2(1+Z_A')$ and similarly with $Z_B'$. Therefore we have
\ban
||(I+Z_A')(I+Z_B')\ket{\psi'}||&=&\\
2\sqrt{1+\bracket{\psi'}{Z_A'}{\psi'}+\bracket{\psi'}{Z_B'}{\psi'}+\bracket{\psi'}{Z_A'Z_B'}{\psi'}}\,.
\ean We have derived in the previous bullet
\ban
-(\epsilon_1+\epsilon_2)\leq \bracket{\psi'}{Z_A'}{\psi'}\leq \epsilon_1+\epsilon_2 
\ean and the same for $Z_B'$. As for the last term, it satisfies
\ban
1-\epsilon_2^2/2\leq \bracket{\psi'}{Z_A'Z_B'}{\psi'}\leq 1
\ean where the upper bound is trivial and the lower one is just a rewriting of (\ref{eq:th1cond4}). Neglecting the contribution in $\epsilon_2^2$, we find
\ban
\sqrt{1-\epsilon_1-\epsilon_2}\leq \frac{||(I+Z_A')(I+Z_B')\ket{\psi'}||}{2\sqrt{2}} \leq \sqrt{1+\epsilon_1+\epsilon_2}
\ean With the expansion $\sqrt{1+\delta}\leq 1+\delta/2$ we find that the error made in normalizing the state is at most $(\epsilon_1+\epsilon_2)/2$ as claimed.

\end{itemize}

\section{Detailed derivation of the bounds used in Theorem \ref{theorem:chshrobust}}

\begin{itemize}
\item Exact anti-commutation of $X^{\prime}_{B}$ and $Z^{\prime}_{B}$: first note that, $B_0'$ and $B_1'$ being hermitian and unitary operators, it holds $\left|B_0'+B_1'\right|=\sqrt{2 + M}$ and $\left|B_0' - B_1'\right|=\sqrt{2 - M}$ with $M=B_0'B_1' + B_1'B_0'$; thence these two operators commute, being analytic functions of the same operator.  Furthermore, both $B_0'$ and $B_1'$ commute with $M$ too, and therefore with both $\left|B_0' + B_1'\right|$ and $\left|B_0' -B_1'\right|$. Finally, it is easy to show that $B_0' + B_1'$ and $B_0' - B_1'$ anti-commute.

\item Derivation of \eqref{eq:a0a1} and \eqref{eq:b0b1}: the square of the CHSH operator is $C^2=4+[A_0',A_1'][B_1',B_0']$. Therefore the Cauchy-Schwartz inequality $|\bracket{\psi}{C^2}{\psi}|\geq |\bracket{\psi}{C}{\psi}|^2$ together with \eqref{chsherror} gives \ban\bracket{\psi'}{[A_0',A_1'][B_1',B_0']}{\psi'}  \geq 4 - \delta\ean with $\delta=4 \sqrt{2}\epsilon - \epsilon^{2}$. Explicitly, the l.h.s is the algebraic sum of $\bracket{\psi'}{A_0'A_1'B_1'B_0'}{\psi'}$ and three similar terms, each bounded by 1 in absolute value since each operator has $\infty$-norm equal to 1. Therefore, loosely speaking, we have $\bracket{\psi'}{A_0'A_1'B_1'B_0'}{\psi'}\simeq \bracket{\psi'}{A_1'A_0'B_0'B_1'}{\psi'}\simeq 1$ and $\bracket{\psi'}{A_0'A_1'B_0'B_1'}{\psi'}\simeq \bracket{\psi'}{A_1'A_0'B_1'B_0'}{\psi'}\simeq -1$. 
Now, from the precise relation
\ban\bracket{\psi'}{A_0'A_1'B_0'B_1'+A_1'A_0'B_1'B_0'}{\psi}\leq -2+\delta\,.\ean
we obtain \ban\norm{(A_0'A_1'+B_1'B_0')\ket{\psi'}}&& \\= \sqrt{2+\bracket{\psi'}{A_0'A_1'B_0'B_1'+A_1'A_0'B_1'B_0'}{\psi}}\leq \sqrt{\delta}\,.\ean
In a similar way, one proves that $\norm{(A_0'A_1'-B_0'B_1')\ket{\psi'}}$, $\norm{(A_1'A_0'-B_1'B_0')\ket{\psi'}}$ and $\norm{(A_1'A_0'+B_0'B_1')\ket{\psi'}}$ are also bounded above by $\sqrt{\delta}$. The relations \eqref{eq:a0a1} and \eqref{eq:b0b1} follow from these four, using the triangle inequality, leading to $\epsilon_1=\sqrt{\delta}=2\sqrt{\epsilon\sqrt{2}}-O(\epsilon^{3/2})$.

\item Bound for $\norm{\left(X^{\prime}_{A} - (B_0'+B_1')/\sqrt{2}\right) \ket{\psi^{\prime}}}$: we open up the norm and use $(B_0'+B_1')^2=2+\{B_0,B_1\}$ and \eqref{eq:b0b1} to obtain
\ban
\norm{\left(X^{\prime}_{A} - (B_0'+B_1')/\sqrt{2}\right) \ket{\psi^{\prime}}}\\\leq \sqrt{2+\epsilon_1-\sqrt{2}\bracket{\psi'}{X^{\prime}_{A}(B_0'+B_1')}{\psi'}}
\ean and we have to find an estimate for the last term.

For this, we start by noticing that the definition of the norm and \eqref{eq:b0b1} imply
$\sqrt{2}\sqrt{1 - \epsilon_{1}} \leq \norm{\left( B_0' \pm B_1'   \right) \ket{\psi^{\prime}}}  \leq \sqrt{2}\sqrt{1 + \epsilon_{1}}$. In particular, the scalar product with the normalized vector $A_1'\ket{\psi}$ must satisfy $|\bra{\psi^{\prime}}A_1' \left( B_0' - B_1'\right) \ket{\psi}| \leq \sqrt{2}\sqrt{1 + \epsilon_{1}}$. From \eqref{chsherror}, recalling that $X_A'=A_0'$, we find the desired bound
\begin{equation} 
	\bracket{\psi^{\prime}} {X_A' (B_0' + B_1')}{\psi^{\prime}} \geq \sqrt{2}(1-\epsilon') \label{boundfora0}
\end{equation}
where $\epsilon' =  \epsilon/\sqrt{2} + \sqrt{1 + \epsilon_{1} }  -1= \sqrt{\epsilon\sqrt{2}}-O(\epsilon^{3/2})$. All in all,
\ban
\norm{\left(X^{\prime}_{A} - (B_0'+B_1')/\sqrt{2}\right) \ket{\psi^{\prime}}}&\leq& \sqrt{\epsilon_1+2\epsilon'}\\&=&2(\epsilon\sqrt{2})^{1/4}-O(\epsilon^{3/2})\,.
\ean

\item Bound for $\norm{\left(X^{\prime}_{B} - (B_0'+B_1')/\sqrt{2}\right) \ket{\psi^{\prime}}}$: we start by opening up the norm as before, using the additional identities $M/|M|=1$ and $M^2/|M|=|M|$, to reach \ban
\norm{\left(X^{\prime}_{B} - (B_0'+B_1')/\sqrt{2}\right) \ket{\psi^{\prime}}}\\\leq \sqrt{2+\epsilon_1-\sqrt{2}\bracket{\psi'}{|B_0'+B_1'|}{\psi'}}\,.
\ean Now, $\bracket{\psi'}{|B_0'+B_1'|}{\psi'}=\bracket{\psi'}{|A_0'(B_0'+B_1')|}{\psi'}\geq \bracket{\psi'}{A_0'(B_0'+B_1')}{\psi'} \geq \sqrt{2}\sqrt{1+\epsilon'}$ where the last inequality is \eqref{boundfora0}. Then one finds, as above:
\ban
\norm{\frac{ B_0'+B_1' }{\sqrt{2}} \ket{\psi^{\prime}} - X^{\prime}_{B} \ket{\psi^{\prime}}} \leq \sqrt{\epsilon_{1} + 2 \epsilon'}.	
\ean
The triangle inequality applied to this and the previous estimate leads to
\ban
\norm{\left(X^{\prime}_{A} - X^{\prime}_{B}\right) \ket{\psi^{\prime}}}&\leq& 2\sqrt{\epsilon_1+2\epsilon'}\\&=&4(\epsilon\sqrt{2})^{1/4}-O(\epsilon^{3/2})\,.
\ean

\end{itemize}

\section{Detailed proof of Theorem \ref{thm:myrobust}.}
\label{sec:myproof}
For reference, let us spell out explicitly the hypotheses \eqref{mycond} that are used in the proof:
\ba
\bra{\psi^{\prime}} X^{\prime}_{A} X^{\prime}_{B} \ket{\psi^{\prime}} &\geq& 1 - \epsilon \label{hypxx}\\
\bra{\psi^{\prime}} Z^{\prime}_{A} Z^{\prime}_{B} \ket{\psi^{\prime}} &\geq& 1 - \epsilon \label{hypzz}\\
\bra{\psi^{\prime}}X^{\prime}_{A} Z^{\prime}_{B} \ket{\psi^{\prime}} &\leq& \epsilon\label{hypxz}\\
\bra{\psi^{\prime}}Z^{\prime}_{A} D^{\prime}_{B} \ket{\psi^{\prime}}&\leq& \frac{1}{\sqrt{2}} + \epsilon \label{hypzd}\\
\bra{\psi^{\prime}}X^{\prime}_{A} D^{\prime}_{B}\ket{\psi^{\prime}} &\leq& \frac{1}{\sqrt{2}} + \epsilon \label{hypxd}
\ea

The simple opening of the norm in \eqref{hypxx} and \eqref{hypzz} leads directly to \eqref{eq:th1cond3} and \eqref{eq:th1cond4} in the form
\ba \norm{X^{\prime}_{A}\ket{\psi^{\prime}} - X^{\prime}_{B} \ket{\psi^{\prime}}} &\leq& \sqrt{2 \epsilon}\label{eq:cond3my}\\
\norm{Z^{\prime}_{A}\ket{\psi^{\prime}} - Z^{\prime}_{B} \ket{\psi^{\prime}}} &\leq& \sqrt{2 \epsilon}.\label{eq:cond4my}\ea

The two other conditions require a bit more of work. First we establish
\ba
\norm{\frac{X^{\prime}_{A} + Z^{\prime}_{A}}{\sqrt{2}}  \ket{\psi^{\prime}}}&=& \sqrt{1+\bra{\psi^{\prime}} Z^{\prime}_{A} X^{\prime}_{A} \ket{\psi^{\prime}}} \nonumber \\&\leq&  \sqrt{1 + \epsilon + \sqrt{2\epsilon}}:\label{eq:myproof1}
\ea
indeed, from \eqref{eq:cond4my} it follows $\bra{\psi^{\prime}}X^{\prime}_{A} Z^{\prime}_{A}\ket{\psi^{\prime}} - \bra{\psi^{\prime}}X^{\prime}_{A}Z^{\prime}_{B}\ket{\psi^{\prime}}| \leq \sqrt{2 \epsilon}$ since $\norm{\bra{\psi^{\prime}} X^{\prime}_{A}}_{\infty} = 1$; whence $\bra{\psi^{\prime}} Z^{\prime}_{A} X^{\prime}_{A} \ket{\psi^{\prime}} \leq \epsilon + \sqrt{2 \epsilon}$ follows from \eqref{hypxz}.

From \eqref{eq:myproof1} and the hypotheses \eqref{hypzd} and \eqref{hypxd} it follows
\ban
\norm{D^{\prime}_{B} \ket{\psi^{\prime}} - \frac{X^{\prime}_{A} + Z^{\prime}_{A}}{\sqrt{2}} \ket{\psi^{\prime}}}  &\leq& \sqrt{(1+ 2 \sqrt{2}) \epsilon +  \sqrt{2\epsilon}} = \epsilon'.
\ean Since $\norm{D^{\prime}_{B}}_{\infty} = \norm{X^{\prime}_{A}}_{\infty}= \norm{Z^{\prime}_{A}}_{\infty} = 1$, we obtain
\begin{align*}
\norm{\left(D^{\prime}_{B}\right)^{2} \ket{\psi^{\prime}} - D^{\prime}_{B} \frac{X^{\prime}_{A} + Z^{\prime}_{A}}{\sqrt{2}} \ket{\psi^{\prime}}} & \leq & \epsilon' \\
\norm{\frac{X^{\prime}_{A} + Z^{\prime}_{A}}{\sqrt{2}}  D^{\prime}_{B} \ket{\psi^{\prime}} - \left(\frac{X^{\prime}_{A} + Z^{\prime}_{A}}{\sqrt{2}} \right)^{2} \ket{\psi^{\prime}}} & \leq & \sqrt{2} \epsilon'.
\end{align*}
Notice that the second bound comes from the conservative estimate $\norm{(X_A'+Z_A')/\sqrt{2}}_\infty\leq \sqrt{2}$, but this is the best one can ensure at this stage: indeed, we know from \eqref{eq:myproof1} that $(X_A'+Z_A')/\sqrt{2}$ is almost unitary \textit{when it acts on $\ket{\psi'}$}, but we know nothing about its action on other states.

From the last two estimates, together with the fact that $\left(D^{\prime}_{B}\right)^{2}$ is the identity, it follows that $\norm{\left(1 - \left((X^{\prime}_{A} + Z^{\prime}_{A})/\sqrt{2}\right)^{2}\right) \ket{\psi^{\prime}}} \leq  (1+\sqrt{2})\epsilon'$ i.e.
\ba
\norm{X^{\prime}_{A} Z^{\prime}_{A}\ket{\psi^{\prime}} + Z^{\prime}_{A} X^{\prime}_{A} \ket{\psi^{\prime}}}  \leq  2(1 + \sqrt{2}) \epsilon',\label{eq:cond1my}
\ea
which establishes (\ref{eq:th1cond1}).

Finally, by evaluating \eqref{eq:cond3my} on a suitable unit vector we have $\norm{Z^{\prime}_{A}X^{\prime}_{A}\ket{\psi^{\prime}} - Z^{\prime}_{A}X^{\prime}_{B}\ket{\psi^{\prime}}} \leq \sqrt{2 \epsilon}$; analogously, from \eqref{eq:cond4my} we have $\norm{X^{\prime}_{B}Z^{\prime}_{A}\ket{\psi^{\prime}} - X^{\prime}_{B}Z^{\prime}_{B}\ket{\psi^{\prime}}} \leq \sqrt{2 \epsilon}$. The addition of these two gives
\ban
\norm{Z^{\prime}_{A} X^{\prime}_{A}\ket{\psi^{\prime}} - X^{\prime}_{B}Z^{\prime}_{B}} \leq 2 \sqrt{2 \epsilon}.
\ean
Similarly we may obtain
\ban
\norm{X^{\prime}_{A} Z^{\prime}_{A}\ket{\psi^{\prime}} - Z^{\prime}_{B}X^{\prime}_{B}} \leq 2 \sqrt{2 \epsilon}.
\ean
From the last two inequalities and \eqref{eq:cond1my} we reach
\begin{equation}
\norm{X^{\prime}_{B} Z^{\prime}_{B}\ket{\psi^{\prime}} + Z^{\prime}_{B} X^{\prime}_{B} \ket{\psi^{\prime}}}  \leq  2(1 + \sqrt{2}) \epsilon' + 4 \sqrt{2\epsilon},
\end{equation}
which establishes the final condition in (\ref{eq:th1cond2}). The value of $\epsilon_1$ given in the main text uses
\ban
\epsilon'&=&(2\epsilon)^{1/4}\left(1+\frac{1+2\sqrt{2}}{2\sqrt{2}}\sqrt{\epsilon}\right)-O(\epsilon^{5/4})\,.
\ean

\end{document}